\documentclass[12pt]{article}
\usepackage{amscd,amsmath,amssymb,amstext,amsthm,exscale,latexsym}
\usepackage{graphicx}
\usepackage{cmap}  \usepackage[T2A]{fontenc}
\usepackage[pdftex,unicode=true,bookmarks=true]{hyperref}
\newtheorem{prop}{Proposition}[section]
\begin{document}
\title     {Global properties of warped solutions in General Relativity with
            electromagnetic field and cosmological constant}
\author    {D. E. Afanasev
            \thanks{E-mail: daniel\_afanasev@yahoo.com}\\
            \sl High school N1561, ul.\ Paustovskogo, 6, kor.\ 2, 117464,
            \sl Moscow\\
            M. O. Katanaev
            \thanks{E-mail: katanaev@mi-ras.ru}\\
            \sl Steklov mathematical institute,
            \sl ul.~Gubkina, 8, Moscow, 119991, Russia}
\date      {20 March 2019}
\maketitle
\begin{abstract}
We consider general relativity with cosmological constant minimally coupled to
electromagnetic field and assume that four-dimensional space-time manifold is
the warped product of two surfaces with Lorentzian and Euclidean signature
metrics. Einstein's equations imply that at least one
of the surfaces must be of constant curvature. It means that the symmetry of the
metric arises as the consequence of equations of motion (``spontaneous
symmetry emergence''). We give classification of global solutions in two cases:
(i) both surfaces are of constant curvature and (ii) the Riemannian surface is
of constant curvature. The latter case includes spherically symmetric solutions
(sphere ${\mathbb S}^2$ with ${\mathbb S}{\mathbb O}(3)$-symmetry group), planar solutions
(two-dimensional Euclidean space ${\mathbb R}^2$ with ${\mathbb I}{\mathbb O}(2)$-symmetry group), and
hyperbolic solutions (two-sheeted hyperboloid ${\mathbb H}^2$ with
${\mathbb S}{\mathbb O}(1,2)$-symmetry). Totally, we get 37 topologically different solutions.
There is a new one among them, which describes changing topology of space in
time already at the classical level.
\end{abstract}
\section{Introduction}
There are many well known exact solutions in general relativity (see, i.e.\
\cite{KrStMaHe80}). To give physical
interpretation of any solution to Einstein's equation, we must know not only the
metric satisfying equations of general relativity but the global structure of
space-time. By this we mean a pair $({\mathbb M},g)$, where ${\mathbb M}$ is the
four-dimensional space-time manifold and $g$ is the metric on ${\mathbb M}$ such, that
manifold ${\mathbb M}$ is maximally extended along geodesics: any geodesic line on ${\mathbb M}$
either can be continued to infinite value of the canonical parameter in both
directions, or it ends up at a singular point, where one of the geometric
invariants becomes infinite. The famous example is the
Kruskal--Szekeres extension \cite{Kruska60,Szeker60} of the Schwarzschild
solutions. In this case, the space-time ${\mathbb M}$ is globally the topological
product of a sphere (spherical symmetry) with the two-dimensional Lorentzian
surface depicted by the well known Carter--Penrose diagram. The knowledge of
this global structure of space-time allows one to introduce the notion of black
and white holes.

The famous Reissner--Nordstr\"om solution \cite{Reissn16,Nordst18}, which is the
spherically symmetric solution of Einstein's equations with electromagnetic
field, is also known globally. There are three types of Carter--Penrose
diagrams: the Reissner-Nordstr\"om black hole, extremal black hole and naked
singularity. The type of the Carter--Penrose diagram depends on the relation
between mass and charge parameters. The spherically symmetric exact solution of
Einstein's equations with electromagnetic field and cosmological constant is
known locally but not analyzed in full detail globally. In this paper, in
particular, we give complete classification of global spherically symmetric
solutions of Einstein's equations with electromagnetic field and cosmological
constant, which depends on relations between three parameters: mass, charge, and
cosmological constant. We show that there are 16 different Carter--Penrose
diagrams in the spherically symmetric case.

In fact, more general classification is given. We do not assume that
solutions have any symmetry from the very beginning. Instead, we require the
space-time to be the warped product of two surfaces: ${\mathbb M}={\mathbb U}\times{\mathbb V}$, where
${\mathbb U}$ and ${\mathbb V}$ are two two-dimensional surfaces with Lorentzian and Euclidean
signature metrics, respectively. As the consequence of the equations of motion,
at least one of the surfaces must be of constant curvature. In this paper, we
consider the cases when (i) both surfaces ${\mathbb U}$ and ${\mathbb V}$ are of constant
curvature and when (ii) only surface ${\mathbb V}$ is of constant curvature. In the
latter case, there are three possibilities: ${\mathbb V}$ is the sphere ${\mathbb S}^2$ (the
spherical ${\mathbb S}{\mathbb O}(3)$ symmetry), the Euclidean plane (the Poincare
${\mathbb I}{\mathbb S}{\mathbb O}(2)$ symmetry), and the two-sheeted hyperboloid ${\mathbb H}^2$ (the
Lorentzian ${\mathbb S}{\mathbb O}(1,2)$ symmetry). We see that the symmetry of solutions is not
assumed from the beginning but arise as the consequence of the equations of
motions. This effect is called ``spontaneous symmetry emergence''. We classify
all global solutions by drawing their Carter--Penrose diagrams for surface
${\mathbb U}$ depending on relations between mass, charge, and cosmological constant.
Totally, there are 4 different Carter--Penrose diagrams in case (i) and 33
globally different solutions in case (ii).

Moreover, we prove that there is the additional forth Killing vector field in
each case. This is a generalization of Birkhoff's theorem stating that any
spherically symmetric solution of vacuum Einstein's equations must be static.
The existence of extra Killing vector field is proved for ${\mathbb S}{\mathbb O}(3)$,
${\mathbb I}{\mathbb S}{\mathbb O}(2)$, and ${\mathbb S}{\mathbb O}(1,2)$ symmetry groups.

This paper follows the classification of global warped product solutions of
general relativity with cosmological constant (without electromagnetic field)
given in \cite{KaKlKu99}. The Carter--Penrose diagrams are constructed using the
conformal block method described in \cite{Katana00A}.

As in \cite{KaKlKu99}, we assume that space-time ${\mathbb M}$ is the warped product
of two surfaces: ${\mathbb M}={\mathbb U}\times{\mathbb V}$, where ${\mathbb U}$ and ${\mathbb V}$ are surfaces with
Lorentzian and Euclidean signature metrics, respectively. Local coordinates on
${\mathbb M}$ are denoted by $x^i$, $i=0,1,2,3$, and coordinates on the surfaces by
Greek letters from the beginning and middle of the alphabet:
\begin{equation*}
  (x^\alpha)\in{\mathbb U},\quad\alpha=0,1,\qquad(y^\mu)\in{\mathbb V},\quad\mu=2,3.
\end{equation*}
That is $(x^i):=(x^\alpha,y^\mu)$. Geometrical notions on four-dimensional
space-time are marked by the hat to distinguish them from notions on surfaces
${\mathbb U}$ and ${\mathbb V}$, which appear more often.

We do not assume any symmetry of solutions from the very beginning.

Four-dimensional metric of the warped product of two surfaces has block diagonal
form by definition:
\begin{equation}                                                  \label{egdtrf}
  \widehat g_{ij}=\begin{pmatrix} k(y)g_{\alpha\beta}(x) & 0 \\ 0 & m(x)h_{\mu\nu}(y)
  \end{pmatrix},
\end{equation}
where $g_{\alpha\beta}(x)$ and $h_{\mu\nu}(y)$ are some metrics on surfaces ${\mathbb U}$ and
${\mathbb V}$, respectively, $k(y)$ and $m(x)$ are scalar (dilaton) fields on
${\mathbb V}$ and ${\mathbb U}$.

The Ricci tensor components for metric (\ref{egdtrf}) are
\begin{equation}                                                  \label{eritab}
\begin{split}
  \widehat R_{\alpha\beta}&=R_{\alpha\beta}+\frac{\nabla_\alpha\nabla_\beta m}m
  -\frac{\nabla_\alpha m\nabla_\beta m}{2m^2}+\frac{g_{\alpha\beta}\nabla^2 k}{2m}
\\
  \widehat R_{\alpha\mu}&=\widehat R_{\mu\alpha}=-\frac{\nabla_\alpha m\nabla_\mu k}{2mk}
\\
 \widehat R_{\mu\nu}&=R_{\mu\nu}+\frac{\nabla_\mu\nabla_\nu k}k
  -\frac{\nabla_\mu k\nabla_\nu k}{2k^2}+\frac{h_{\mu\nu}\nabla^2 m}{2k},
\end{split}
\end{equation}
where, for brevity, we introduce notation
\begin{equation}                                                  \label{edalan}
  \nabla^2 m:=g^{\alpha\beta}\nabla_\alpha\nabla_\beta m,\qquad
  \nabla^2 k:=h^{\mu\nu}\nabla_\mu\nabla_\nu k.
\end{equation}
Here and in what follows symbol $\nabla$ denotes covariant derivative with the
corresponding Christoffel's symbols. The four-dimensional scalar curvature is
\begin{equation}                                                  \label{escacu}
  \widehat R=\frac1kR^g+2\frac{\nabla^2 m}{km}-\frac{(\nabla m)^2}{2km^2}
  +\frac1mR^h+2\frac{\nabla^2 k}{km}-\frac{(\nabla k)^2}{2k^2m},
\end{equation}
where
\begin{equation}                                                  \label{egrsqn}
  (\nabla m)^2:=g^{\alpha\beta}\partial_\alpha m\partial_\beta m,\qquad
  (\nabla k)^2:=h^{\mu\nu}\partial_\mu k\partial_\nu k.
\end{equation}
Scalar curvatures of surfaces ${\mathbb U}$ and ${\mathbb V}$ are denoted by $R^g$ and $R^h$,
respectively.
\section{Solution for electromagnetic field}
We assume that electromagnetic field is minimally coupled to gravity. Then the
action takes the form
\begin{equation}                                                  \label{ubcftg}
  S=\int\!dx\sqrt{|\widehat g|}\left(\widehat R-2\Lambda-\frac14\widehat F^2\right),
\end{equation}
where $\widehat R$ is the scalar curvature for metric $\widehat g_{ij}$,
$\widehat g:=\det\widehat g_{ij}$, $\Lambda$ is a cosmological constant, and
$\widehat F^2$ is the square of electromagnetic field strength:
\begin{equation*}
  \widehat F^2:=\widehat F_{ij}\widehat F^{ij},\qquad
  \widehat F_{ij}:=\partial_i \widehat A_j-\partial_j\widehat A_i.
\end{equation*}
Here, $\widehat A_i$ are components of electromagnetic field potential. For
brevity, gravitational and electromagnetic coupling constants are set to unity.

Variation of action (\ref{ubcftg}) with respect to metric yields
four-dimensional Einstein's equations:
\begin{equation}                                                  \label{ubsvgr}
  \widehat R_{ij}-\frac12\widehat g_{ij}\widehat R+\widehat g_{ij}\Lambda
  =-\frac12 \widehat T_{{\textsc{e}}{\textsc{m}} ij},
\end{equation}
where
\begin{equation}                                                  \label{ubnhgy}
  \widehat T_{{\textsc{e}}{\textsc{m}} ij}:=-\widehat F_{ik}\widehat F_j{}^k
  +\frac14\widehat g_{ij}\widehat F^2
\end{equation}
is the electromagnetic field energy-momentum tensor. Variation of
the action with respect to electromagnetic field yields Maxwell's equations:
\begin{equation}                                                  \label{uncmhy}
  \partial_j\big(\sqrt{|\widehat g|}\widehat F^{ji}\big),
\end{equation}
where
\begin{equation*}
  \widehat g=k^2m^2gh,\qquad g:=\det g_{\alpha\beta},\qquad h:=\det h_{\mu\nu}.
\end{equation*}

To simplify the problem, we assume that the four-dimensional electromagnetic
potential consists of two parts:
\begin{equation*}
  \widehat A_i=\big(A_\alpha(x),A_\mu(y)\big),
\end{equation*}
where $A_\alpha(x)$ and $A_\mu(y)$ are two-dimensional electromagnetic potentials
on surfaces ${\mathbb U}$ and ${\mathbb V}$, respectively. Then the electromagnetic field
strength becomes block diagonal:
\begin{equation}                                                  \label{unfhgt}
  \widehat F_{ij}=\begin{pmatrix}F_{\alpha\beta} & 0 \\ 0 & F_{\mu\nu} \end{pmatrix},
\end{equation}
where
\begin{equation*}
  F_{\alpha\beta}(x):=\partial_\alpha A_\beta-\partial_\beta A_\alpha,\qquad
  F_{\mu\nu}(y):=\partial_\mu A_\nu-\partial_\nu A_\mu
\end{equation*}
are strength components for two-dimensional electromagnetic potentials.

In what follows, the raising of Greek indices from the beginning and middle
of the Greek alphabet is performed by using the inverse metrics $g^{\alpha\beta}$ and
 $h^{\mu\nu}$. Therefore
\begin{equation*}
  \widehat F^{\alpha\beta}=\frac1{k^2}F^{\alpha\beta},\qquad
  \widehat F^{\mu\nu}=\frac1{m^2}F^{\mu\nu},
\end{equation*}
where $k(y)$ and $m(x)$ are dilaton fields entering four-dimensional metric
(\ref{egdtrf}). The square of four-dimensional electromagnetic field strength is
\begin{equation*}
  \widehat F^2=\frac1{k^2}F_{\alpha\beta}F^{\alpha\beta}+\frac1{m^2}F_{\mu\nu}F^{\mu\nu}.
\end{equation*}

In the case under consideration, Maxwell's Eqs.~(\ref{uncmhy}) for $i=\alpha$ lead
to equality
\begin{equation*}
  \frac1{|k|}\sqrt{|h|}\partial_\beta\left(|m|\sqrt{|g|}F^{\beta\alpha}\right)=0.
\end{equation*}
A general solution to these equations has the form
\begin{equation}                                                  \label{uncbfh}
  |m|\sqrt{|g|}F^{\alpha\beta}=2\hat\varepsilon^{\alpha\beta}Q,\qquad Q={\sf\,const},
\end{equation}
where $\hat\varepsilon^{\alpha\beta}$ is the totally antisymmetric second rank tensor
density. The factor 2 is introduced in the right hand side of general solution
for simplification of subsequent formulae. This solution is rewritten as
\begin{equation}                                                  \label{ubvcfr}
  F^{\alpha\beta}=\frac{2Q}{|m|}\varepsilon^{\alpha\beta},
\end{equation}
where $\varepsilon^{\alpha\beta}:=\hat\varepsilon^{\alpha\beta}/\sqrt{|g|}$ is now the totally
antisymmetric second rank tensor.

If $i=\mu$, then Maxwell's Eqs.~(\ref{uncmhy}) yield the equality
\begin{equation*}
  \frac1{|m|}\sqrt{|g|}\partial_\mu\left(|k|\sqrt h F^{\mu\nu}\right)=0.
\end{equation*}
Its general solution is
\begin{equation}                                                  \label{unbgtm}
  F^{\mu\nu}=\frac{2P}{|k|}\varepsilon^{\mu\nu},\qquad P={\sf\,const}.
\end{equation}

Now the four-dimensional electromagnetic energy-momentum tensor (\ref{ubnhgy})
is easily calculated. It is block diagonal:
\begin{equation}                                                  \label{uvbxfp}
  \widehat T_{ij}=\begin{pmatrix}
  \widehat T_{\alpha\beta} & 0 \\ 0 & \widehat T_{\mu\nu} \end{pmatrix},
\end{equation}
where
\begin{equation*}
  \widehat T_{\alpha\beta}=\frac{2g_{\alpha\beta}}{km^2}(Q^2+P^2),\qquad
  \widehat T_{\mu\nu}=-\frac{2h_{\mu\nu}}{k^2m}(Q^2+P^2).
\end{equation*}

Note that we do not need the electromagnetic potentials $A_\alpha$ and $A_\mu$ for
the calculation of the energy-momentum tensor. It is sufficient to know
strengthes (\ref{ubvcfr}) and (\ref{unbgtm}).

Now we have to solve Einstein's Eqs.~(\ref{ubsvgr}) with right hand side
(\ref{uvbxfp}). Since energy-momentum tensor depends only on the sum $Q^2+P^2$,
we set $P=0$ to simplify formulae. In the final answer, this constant is easily
reconstructed by substitution $Q^2\mapsto Q^2+P^2$.

In what follows, we consider only the case $Q\ne0$, because the case $Q=0$ was
considered in \cite{KaKlKu99} in full detail.
\section{Einstein's equations}
The right hand side of Einstein's Eqs.~(\ref{ubsvgr}) is defined by general
solution of Maxwell's equations, which leads to electromagnetic energy-momentum
tensor (\ref{uvbxfp}). The trace of Einstein's equations can be easily solved
with respect to the scalar curvature:
\begin{equation*}
  \widehat R=4\Lambda,
\end{equation*}
which does not depend of the electromagnetic field, because the trace of the
electromagnetic field energy-momentum tensor equals zero. After elimination of
the scalar curvature, Einstein's equations are simplified:
\begin{equation}                                                  \label{ubbdgt}
  \widehat R_{ij}-\widehat g_{ij}\Lambda=-\frac12\widehat T_{{\textsc{e}}{\textsc{m}} ij}.
\end{equation}
For indices values $(ij)=(\alpha,\beta)$, $(\mu\nu)$, and $(\alpha,\mu)$, these
equations yield the following system of equations:
\begin{align}                                                     \label{ubvcfl}
  R_{\alpha\beta}+\frac{\nabla_\alpha\nabla_\beta m}m-\frac{\nabla_\alpha m\nabla_\beta m}{2m^2}
  +g_{\alpha\beta}\left(\frac{\nabla^2k}{2m}-k\Lambda+\frac{Q^2}{m^2k}\right)=&0,
\\                                                                \label{ubncjh}
  R_{\mu\nu}+\frac{\nabla_\mu\nabla_\nu k}k-\frac{\nabla_\mu k\nabla_\nu k}{2k^2}
  +h_{\mu\nu}\left(\frac{\nabla^2 m}{2k}-m\Lambda-\frac{Q^2}{k^2m}\right)=&0,
\\
  -\frac{\nabla_\alpha m\nabla_\mu k}{2mk}=&0,
\end{align}
where $R_{\alpha\beta}$ and $R_{\mu\nu}$ are Ricci tensors for two-dimensional
metrics $g_{\alpha\beta}$ and $h_{\mu\nu}$, respectively, $\nabla_\alpha$ and $\nabla_\mu$ are
two-dimensional covariant derivatives with Christoffel's symbols on surfaces
${\mathbb U}$ and ${\mathbb V}$, $\nabla^2:=g^{\alpha\beta}\nabla_\alpha\nabla_\beta$ or
$\nabla^2:=h^{\mu\nu}\nabla_\mu\nabla_\nu$, which is clear from the context. Sure, the
equalities $\nabla_\alpha m=\partial_\alpha m$ and $\nabla_\mu k=\partial_\mu k$ hold. But we keep
the symbol of covariant derivative for uniformity.

For subsequent analysis of Einstein's equations, we extract the traces and
traceless parts from Eqs.~\ (\ref{ubvcfl}) and (\ref{ubncjh}). Then the full
system of Einstein's equations takes the form
\begin{align}                                                     \label{unncbg}
  \nabla_\alpha\nabla_\beta m-\frac{\nabla_\alpha m\nabla_\beta m}{2m}-\frac12\left(\nabla^2 m
  -\frac{(\nabla m)^2}{2m}\right)=&0,
\\                                                                \label{ubmsdi}
  \nabla_\mu\nabla_\nu k-\frac{\nabla_\mu k\nabla_\nu k}{2k}-\frac12\left(\nabla^2 k
  -\frac{(\nabla k)^2}{2k}\right)=&0,
\\                                                                \label{undbyt}
  R^g+\frac{\nabla^2 m}m-\frac{(\nabla m)^2}{2m^2}+\frac{\nabla^2 k}m-2k\Lambda
  +\frac{2Q^2}{m^2k}=&0,
\\                                                                \label{undhtt}
  R^h+\frac{\nabla^2 k}k-\frac{(\nabla k)^2}{2k^2}+\frac{\nabla^2 m}k-2m\Lambda
  -\frac{2Q^2}{k^2m}=&0,
\\                                                                \label{ubvfds}
  \nabla_\alpha m\nabla_\beta k=&0,
\end{align}
where $(\nabla m)^2:=g^{\alpha\beta}\nabla_\alpha m\nabla_\beta m$,
$(\nabla k)^2:=g^{\mu\nu}\nabla_\mu k\nabla_\nu k$, $R^g$ and $R^h$ are scalar curvatures
of two-dimensional surfaces ${\mathbb U}$ and ${\mathbb V}$ for metrics $g$ and $h$,
respectively. In the above formulae, we used equalities
$R_{\alpha\beta}=\frac12g_{\alpha\beta}R^g$ and $R_{\mu\nu}=\frac12h_{\mu\nu}R^h$ valid in
two dimensions.

The last Eq.~(\ref{ubvfds}), which corresponds to mixed values of indices
$(ij)=(\alpha\mu)$ in Einstein's equations results in strong restrictions on
solutions. Namely, as in the case without electromagnetic field, there are only
three cases:
\begin{equation}                                                  \label{ecasek}
\begin{array}{lrr}
  {\sf A}:  & \qquad k={\sf\,const}\ne0,      & \qquad m={\sf\,const}\ne0, \\
  {\sf B}:  & k={\sf\,const}\ne0,      & \nabla_\alpha m\ne0, \\
  {\sf C}:  & \nabla_\mu k\ne0,  & m={\sf\,const}\ne0.
\end{array}
\end{equation}
We shall see in what follows, that this leads to ``spontaneous symmetry
emergence''.

Now we consider the first two cases in detail.
\section{Product of constant curvature surfaces}
The most symmetric solutions of Einstein's equations with electromagnetic field
in the form of the product of two constant curvature surfaces arise in case
{\sf A} (\ref{ecasek}), when both dilaton fields are constant. If $k$ and $m$
are constant, then Eqs.~(\ref{unncbg}) and (\ref{ubmsdi}) are identically
satisfied, and Eqs.~(\ref{undbyt}) and (\ref{undhtt}) take the form
\begin{equation}                                                  \label{ubbcnd}
  R^q=2k\Lambda-\frac{2Q^2}{m^2k}=-2K^g,\qquad R^h=2m\Lambda+\frac{2Q^2}{k^2m}
  =-2K^h,
\end{equation}
where
\begin{equation*}
  K^g:=-k\left(\Lambda-\frac{Q^2}{k^2m^2}\right),\qquad
  K^h:=-m\left(\Lambda+\frac{Q^2}{k^2m^2}\right)
\end{equation*}
are Gaussian curvatures of surfaces ${\mathbb U}$ and ${\mathbb V}$, respectively. It means that
both surfaces are of constant curvature in case {\sf A}. The metric on each
surface is invariant under three-dimensional transformation group.

In stereographic coordinates on both surfaces, the metric of four-dimensional
space-time takes the form
\begin{equation}                                                  \label{ubvxgh}
\begin{split}
  ds^2=&k g_{\alpha\beta}dx^\alpha dx^\beta+mh_{\mu\nu}dy^\mu dy^\nu=
\\[6pt]
  =&k\frac {dt^2-dx^2}{\big[1+\frac{K^g}4(t^2-x^2)\big]^2}
  +m\frac{dy^2+dz^2}{\big[1+\frac{K^h}4(y^2+z^2)\big]^2},
\end{split}
\end{equation}
where $(x^\alpha):=(t,x)$ and $(y^\mu):=(y,z)$.

We can put $k=\pm1$ and $m=\pm1$ by rescaling coordinates. One has also to
redefine the constant of integration $Q^2/(k^2m^2)\mapsto Q^2$. We choose $k=1$
and $m=-1$ for the metric signature to be $(+---)$. Then the Gaussian curvatures
are
\begin{equation}                                                  \label{umnkiu}
  K^g=Q^2-\Lambda,\qquad K^h=Q^2+\Lambda.
\end{equation}
There are four qualitatively different cases for topologically inequivalent
global solutions depending on relations between cosmological constant and
charge:
\begin{equation}                                                  \label{ubsndi}
\begin{aligned}
  \Lambda<-Q^2: & \qquad K^g>0, & K^h<0, & \qquad {\mathbb M}={\mathbb L}^2\times{\mathbb H}^2,
\\
  \Lambda=-Q^2: & \qquad K^g>0, & K^h=0, & \qquad {\mathbb M}={\mathbb L}^2\times{\mathbb R}^2,
\\
  -Q^2<\Lambda<Q^2: & \qquad K^g>0, & K^h>0, & \qquad {\mathbb M}={\mathbb L}^2\times{\mathbb S}^2,
\\
  \Lambda=Q^2: & \qquad K^g=0, & K^h>0, & \qquad {\mathbb M}={\mathbb R}^{1,1}\times{\mathbb S}^2,
\\
  \Lambda>Q^2: & \qquad K^g<0, & K^h>0, & \qquad {\mathbb M}={\mathbb L}^2\times{\mathbb S}^2,
\end{aligned}
\end{equation}
where ${\mathbb L}^2$ is the one sheet hyperboloid (more precisely, its universal
covering) embedded in three-dimensional Minkowskian space ${\mathbb R}^{1,2}$,
${\mathbb H}^2$ is the Lobachevsky plane (the upper sheet of two-sheeted hyperboloid
embedded in ${\mathbb R}^{1,2}$), and ${\mathbb S}^2$ is the two-dimensional sphere. From
topological point of view the third and fifth cases in Eq.~(\ref{ubsndi})
coincide. Therefor there are only four topologically inequivalent global
solutions of Einstein's equations in the form of direct product of two constant
curvature surfaces. Note that for $Q=0$, there are only three topologically
inequivalent solutions \cite{KaKlKu99}.

All solutions have exactly six Killing vector fields and belong to type $D$ in
Petrov's classification.

The cases of other signatures of four-dimensional metric for $k=\pm1$ and
$m=\pm1$ are analysed similarly. Qualitative properties of global solutions are
the same.

We see that symmetry properties in this case are not imposed from the very
beginning but arise as the result of solution of equations of motion. This
effect is called ``spontaneous symmetry emergence''.
\section{Solutions with spatial symmetry}
The dilaton field $k$ is constant in second case {\sf B} (\ref{ecasek}).
Without loss of generality, we put $k=1$. Then Einstein's equations
(\ref{unncbg})--(\ref{ubvfds}) take the form
\begin{align}                                                     \label{ubvhhj}
  \nabla_\alpha\nabla_\beta m-\frac{\nabla_\alpha m\nabla_\beta m}{2m}
  -\frac12g_{\alpha\beta}\left[\nabla^2 m-\frac{(\nabla m)^2}{2m}\right]=&0,
\\                                                                \label{undmwi}
  R^h+\nabla^2 m-2m\Lambda-\frac{2Q^2}m=&0,
\\                                                                \label{usneoo}
  R^g+\frac{\nabla^2 m}m-\frac{(\nabla m)^2}{2m^2}-2\Lambda+\frac{2Q^2}{m^2}=&0.
\end{align}

Consider Eq.~(\ref{undmwi}). The scalar curvature $R^h$ depends on coordinates
$y$ on surface ${\mathbb V}$, whereas all other terms depend on coordinates $x$ on
surface ${\mathbb U}$. For this equation to be fulfilled, it is necessary that
equation $R^h={\sf\,const}$ holds. It means that surface ${\mathbb V}$ must be of constant
curvature as the consequence of Einstein equations. Therefor the
four-dimensional metric of space-time has at least three independent Killing
vector fields. So, there is spontaneous symmetry emergence.

Let us put $R^h:=-2K^h={\sf\,const}$. Then Eq.~(\ref{undmwi}) is
\begin{equation}                                                  \label{unnvbh}
  \nabla^2m-2m\Lambda-2K^h-\frac{2Q^2}m=0.
\end{equation}

Excluding the case {\sf A} considered in the previous section, we proceed
further assuming $\nabla_\alpha m\ne0$ on the whole ${\mathbb U}$.
\begin{prop}                                                      \label{pkfgju}
Equation (\ref{unnvbh}) is the first integral of Eqs.~(\ref{ubvhhj}) and
(\ref{usneoo}).
\end{prop}
\begin{proof}
Differentiate Eq.~(\ref{unnvbh}) and use the equality
\begin{equation*}
  [\nabla_\alpha,\nabla_\beta]A_\gamma=-R^g_{\alpha\beta\gamma}{}^\delta A_\delta,
\end{equation*}
valid for any covector field $A_\alpha$, to change the order of
derivatives in the first term:
\begin{multline*}
  \nabla_\alpha(\ref{unnvbh})=\frac{\nabla^\beta m\nabla_\alpha m\nabla_\beta m}{2m}
  +\frac{\nabla_\alpha m\nabla^2m}{2m}-\frac{\nabla_\alpha m(\nabla m)^2}{2m^2}+\\
\\
  +\frac12\nabla_\alpha\left(\nabla^2 m-\frac{(\nabla m)^2}{2m}\right)
  +\frac12\nabla_\alpha m R^g-2\nabla_\alpha m\Lambda+\nabla_\alpha m\frac{2Q^2}{m^2}.
\end{multline*}
Now exclude derivatives $\nabla^\beta m\nabla_\alpha m$ and $\nabla^2 m$ using
Eqs.~(\ref{ubvhhj}) and (\ref{undmwi}) in the first and fourth terms on the
right hand side. After rearranging terms, the sum of the first and fourth terms
takes the form
\begin{equation*}
  \nabla_\alpha m\left(\frac{(\nabla m)^2}{4m^2}+\Lambda-\frac{Q^2}{m^2}\right).
\end{equation*}
Taking all terms together, we obtain
\begin{equation}                                                  \label{uncbvg}
  \nabla_\alpha(\ref{unnvbh})=\frac12\nabla_\alpha m(\ref{usneoo}).
\end{equation}
Since $\nabla_\alpha m\ne0$, it implies the statement of the proposition.
\end{proof}

The proof of the proposition implies that it is sufficient to solve
Eqs.~(\ref{ubvhhj}) and (\ref{unnvbh}), Eq.~(\ref{usneoo}) being satisfied
automatically.

To solve Eqs.~(\ref{ubvhhj}) and (\ref{unnvbh}) explicitly, we fix the
conformal gauge for metric $g_{\alpha\beta}$ on Lorentzian surface ${\mathbb U}$:
\begin{equation}                                                  \label{endhft}
  g_{\alpha\beta}dx^\alpha dx^\beta=\Phi d\xi d\eta,
\end{equation}
where $\Phi(\xi,\eta)\ne0$ is the conformal factor depending on light cone
coordinates $\xi:=\tau+\sigma$, $\eta:=\tau-\sigma$ on ${\mathbb U}$. The respective four
dimensional metric is
\begin{equation}                                                  \label{unjhsy}
  ds^2=\Phi d\xi d\eta+md\Omega,
\end{equation}
where $d\Omega$ is the metric on the Riemannian surface of constant curvature
${\mathbb V}={\mathbb S}^2$, ${\mathbb R}^2$, or ${\mathbb H}^2$. The sign of the conformal factor $\Phi$ is not
fixed for the present.

For $\Phi>0$ and $m<0$ the signature of metric (\ref{unjhsy}) is $(+---)$. If we
change the sign of $m$, the signature of the metric becomes $(+-++)$. The same
transformation of the signature can be achieved by changing the overall sign of
the metric $\widehat g_{ij}\mapsto-\widehat g_{ij}$, and interchanging the
first two coordinates, $\tau\leftrightarrow\sigma$. Einstein's equations with
cosmological constant and electromagnetic field (\ref{ubbdgt}) are not invariant
with respect to these transformations with simultaneous changing the sign of the
cosmological constant, because the right hand side changes its sign. Therefor,
for $\Phi>0$, we have to consider two cases:
\begin{equation*}
  m<0\quad\Leftrightarrow\quad{\sf\,sign\,}\widehat g_{ij}=(+---)\qquad\text{и}\qquad
  m>0\quad\Leftrightarrow\quad{\sf\,sign\,}\widehat g_{ij}=(-+++).
\end{equation*}
This is the difference for Einstein's equations without electromagnetic field
considered in \cite{KaKlKu99}.
\subsubsection{Metric signature $(+---)$}
For $\Phi>0$ and $m<0$, we introduce convenient parameterization
\begin{equation}                                                  \label{ubndoi}
  m:=-q^2,\qquad q(\xi,\eta)>0.
\end{equation}
Afterwards, we obtain the full system of equations:
\begin{align}                                                     \label{egcpga}
  -\partial^2_{\xi\xi}q+\frac{\partial_\xi\Phi\partial_\xi q}\Phi&=0,
\\                                                                \label{egcygb}
  -\partial^2_{\eta\eta}q+\frac{\partial_\eta\Phi\partial_\eta q}\Phi&=0,
\\                                                                \label{egcugc}
  -2\frac{\partial^2_{\xi\eta}q^2}\Phi-K^h+\Lambda q^2+\frac{Q^2}{q^2}&=0.
\end{align}
The first two equations which do not depend on the electromagnetic field imply
the following assertion.
\begin{prop}
If $\partial_\xi q\partial_\eta q>0$, then the function $q(\tau)$ depends only on
timelike coordinate $\tau:=\frac12(\xi+\eta)$. If $\partial_\xi q\partial_\eta q<0$, then
the function $q(\sigma)$ depends only on spacelike coordinate
$\sigma:=\frac12(\xi-\eta)$. And the following equality holds
\begin{equation}                                                  \label{unbcgt}
  |\Phi|=|q'|,
\end{equation}
where prime denotes differentiation on the argument (either $\tau$, or $\sigma$).
\end{prop}
This proposition provides a general solution to equations (\ref{egcpga}) and
(\ref{egcygb}) up to conformal transformations. This statement is proved in
\cite{KaKlKu99,Katana13B}.

Thus, we can always choose coordinates in such a way that $q$ and $\Phi$
depend simultaneously on timelike or spacelike coordinate
\begin{equation}                                                  \label{eindvj}
  \zeta:=\frac12(\xi\pm\eta)=:\begin{cases} \tau,\qquad \partial_\xi q\,\partial_\eta q>0,
  \\ \sigma,\qquad \partial_\xi q\,\partial_\eta q<0. \end{cases}
\end{equation}
It means that two-dimensional metric (\ref{endhft}) and consequently
four-dimensional metric (\ref{unjhsy}) have the Killing vector $\partial_\sigma$ or
$\partial_\tau$, as the consequence of equations (\ref{egcpga}) and (\ref{egcygb}).
We call these solutions homogeneous and static, respectively, though it is
related to the fixed coordinate system. The existence of additional Killing
vector is the generalization of Birkhoff's theorem \cite{Birkho23} stating that
arbitrary spherically symmetric solution of vacuum Einstein's equations must be
static. (This statement was previously published in \cite{Jebsen21}.)
The generalization includes the addition of electromagnetic field, and, in
addition, the existence of extra Killing vector is proved not only for
spherically symmetric solution $(K^h=1)$, but also for solutions invariant with
respect to ${\mathbb I}{\mathbb S}{\mathbb O}(2)$ $(K^h=0)$ and ${\mathbb S}{\mathbb O}(1,2)$ $(K^h=-1)$ transformation
groups.

We are left to solve equation (\ref{egcugc}). In static, $q=q(\sigma)$, and
homogeneous, $q=q(\tau)$, cases, equation (\ref{egcugc}) takes the form
\begin{align}                                                     \label{ufijhg}
  (q^2)''&=~~2\left(K^h-\Lambda q^2-\frac{Q^2}{q^2}\right)\Phi, & q&=q(\sigma),
\\                                                                \label{usedrf}
  (q^2)''&=-2\left(K^h-\Lambda q^2-\frac{Q^2}{q^2}\right)\Phi, & q&=q(\tau).
\end{align}
To integrate the derived equations, one has to express $\Phi$ through $q$ using
equation (\ref{unbcgt}) and removing the modulus sign.

We consider the static case $q=q(\sigma)$, $\Phi>0$ and $q'>0$ in detail. Then
Eq.~(\ref{ufijhg}) together with Eq.~(\ref{unbcgt}) reduces to
\begin{equation*}
  (q^2)''=2\left(K^h-\Lambda q^2-\frac{Q^2}{q^2}\right)q'.
\end{equation*}
It can be easily integrated:
\begin{equation*}
  (q^2)'=2\left(K^hq-\frac{\Lambda q^3}3-2M+\frac{Q^2}q\right),
\end{equation*}
where $M={\sf\,const}$ is an integration constant, which coincides with mass in the
Schwarzschild solution. Differentiating the left hand side and dividing it by
$2q>0$, we obtain equation
\begin{equation*}
  q'=K^h-\frac{2M}q+\frac{Q^2}{q^2}-\frac{\Lambda q^2}3.
\end{equation*}
Since $q'=\Phi$ in the case under consideration, it implies expression for the
conformal factor through variable $q$:
\begin{equation}                                                  \label{ucohjy}
  \Phi(q)=K^h-\frac{2M}q+\frac{Q^2}{q^2}-\frac{\Lambda q^2}3.
\end{equation}

If $q=q(\sigma)$, $\Phi>0$ and $q'<0$, then the similar integration yields
\begin{equation*}
  q'=-\Phi(q),
\end{equation*}
where the same conformal factor (\ref{ucohjy}) stands in the right hand side.
This case can be united with the previous one by re-writing equation for $q$ in
the form
\begin{equation}                                                  \label{ueqsra}
  |q'|=\Phi(q),\qquad q=q(\sigma),\quad \Phi>0.
\end{equation}
The modulus sign in the left hand side means that if $q(\sigma)$ is a solution, then
the function $q(-\sigma)$ is also the solution.

The static case for $\Phi<0$ is integrated in the same way:
\begin{equation}                                                  \label{uftrye}
  |q'|=-\Phi(q),\qquad q=q(\sigma),\quad \Phi<0.
\end{equation}

If solution is homogeneous, $q=q(\tau)$ and $\Phi>0$, $q'>0$, then integration
of Eq.~(\ref{usedrf}) yields the equality
\begin{equation*}
  q'=-\left(K^h-\frac{2M}q+\frac{Q^2}{q^2}-\frac{\Lambda q^2}3\right).
\end{equation*}
That is the conformal factor must be identified with the right hand side
\begin{equation}                                                  \label{unewsd}
  \hat\Phi=-\left(K^h-\frac{2M}q+\frac{Q^2}{q^2}-\frac{\Lambda q^2}3\right).
\end{equation}
We denote the expression for the conformal factor through $q$ by hat because in
homogeneous case it differs by the sign. Thus, homogeneous solutions of
Einstein's equations can be written in the form
\begin{alignat}{3}                                                \label{ufredl}
  |q'|&=~~\hat\Phi(q),\qquad & q&=q(\tau),\quad && \hat\Phi>0.
\\                                                                \label{usdeax}
  |q'|&=-\hat\Phi(q),\qquad & q&=q(\tau),\quad && \hat\Phi<0.
\end{alignat}

If the conformal factor is negative, then the signature of the metric is
$(-+--)$. In this case, we return to the previous signature $(+---)$ after
substitution $\tau\leftrightarrow\sigma$. This transformation allows us to unite
static and homogeneous solutions by taking the modulus of the conformal factor
in the expression for metric (\ref{unjhsy}). Then a general solution of vacuum
Einstein's equations with electromagnetic field (\ref{ubsvgr}) in the
corresponding coordinate system takes the form
\begin{equation}                                                  \label{qnhtsd}
 ds^2=|\Phi|(d\tau^2-d\sigma^2)-q^2d\Omega,
\end{equation}
where the conformal factor $\Phi$ is given by Eq.~(\ref{ucohjy}). Here the
variable $q$ depends on $\sigma$ (static local solution) or $\tau$
(homogeneous local solution) through the differential equation
\begin{equation}                                                  \label{qdtres}
  \left|\frac{dq}{d\zeta}\right|=\pm\Phi(q),
\end{equation}
where the sign rule holds:
\begin{equation}                                                  \label{esignk}
\begin{array}{ccl}
\Phi>0: & \quad \zeta=\sigma, &\quad \text{the sign $+$ (static local solution)},\\
\Phi<0: & \quad \zeta=\tau,&\quad \text{the sign $-$ (homogeneous local solution)}.
\end{array}
\end{equation}
Thus the four-dimensional Einstein's equations imply that there is the metric
with one Killing vector field on surface ${\mathbb U}$ which was considered in full
detail in \cite{Katana00A}. Now we can construct global solutions (maximally
extended along geodesics) of vacuum Einstein's equations using the conformal
block method. The number of singularities and zeroes of conformal factor
(\ref{ucohjy}) depends on relations between constants $K$, $M$, $Q$, and $\Lambda$.
Therefor there are many qualitatively different global solutions, which are
considered in next sections.

Conformal factor (\ref{ucohjy}) has one singularity: the second order pole at
$q=0$. Therefor according to the rules formulated in \cite{Katana00A,Katana13B}
every global solution correspond to one of the intervals $(-\infty,0)$ or
$(0,\infty)$. The form of conformal factor (\ref{ucohjy}) implies that
these global solutions are obtained one from the other by the transformation
$M\mapsto-M$. Hence, without loss of generality, we describe global solutions
corresponding to both intervals but positive values of $M$.

Because conformal factor $\Phi(q)$ is a smooth function for $q\ne0$, all
arising Lorentzian surfaces ${\mathbb U}$ and metrics on them are smooth.

To conclude the section we compute geometrical invariants which show that
obtained solution of Einstein's equations are nontrivial. First, we compute
the scalar curvature $R^g$ of the surface ${\mathbb U}$. Equations (\ref{undmwi}) and
(\ref{usneoo}) imply
\begin{equation*}
  R^g=-\frac{2K^h}m+\frac{(\nabla m)^2}{2m^2}-\frac{4Q^2}{m^2}
  =\frac{2K^h}{q^2}+\frac{2(\nabla q)^2}{q^2}-\frac{4Q^2}{q^4}.
\end{equation*}
Since
\begin{equation*}
  (\nabla q)^2=\frac1\Phi\eta^{\alpha\beta}\partial_\alpha q\partial_\beta q
  =-\frac{q'{}^2}\Phi
\end{equation*}
both for static and homogeneous solutions, the final expression is
\begin{equation}                                                  \label{ubjdgh}
  R^g=\frac{2\Lambda}3+\frac{4M}{q^3}-\frac{6Q^2}{q^4}.
\end{equation}
It does not depend on Gaussian curvature $K^h$ of Riemannian surface ${\mathbb V}$ and
is singular for $q=0$ if $M\ne0$ and/or $Q\ne0$.
\subsubsection{Metric signature $(-+++)$}
If $m>0$, then the signature of the metric is opposite $(-+++)$, and we
introduce parameterization
\begin{equation*}
  m:=q^2,\qquad q>0,
\end{equation*}
instead of Eq.~(\ref{ubndoi}). Performing the same calculation as in the previous
section, we obtain the first order equation for $q$:
\begin{equation}                                                  \label{uncmxy}
  \left|\frac{dq}{d\zeta}\right|=\pm\Phi(q),
\end{equation}
where $M$ is an integration constant and
\begin{equation}                                                  \label{ubbchd}
  \Phi(q):=\left(K^h-\frac{2M}q-\frac{Q^2}{q^2}+\frac{\Lambda q^2}3\right).
\end{equation}
Here we must take into account that for getting the signature $(-+++)$ we have
to make interchanging $\tau\leftrightarrow\sigma$. We see that for drawing the
Carter--Penrose diagram one has to simply make replacement $Q^2\mapsto-Q^2$ and
$\Lambda\mapsto-\Lambda$ as compared to signature $(+---)$.

Now we describe all spatially symmetric global solution of Einstein's equations
with electromagnetic field which are defined by zeroes and their types of the
conformal factor $\Phi(q)$.
\subsection{Spherically symmetric solutions $K^h=1$              \label{sphers}}
In the considered case, global spherically symmetric solutions, that is pairs
$({\mathbb M},\widehat g)$, have the form ${\mathbb M}={\mathbb U}\times{\mathbb S}^2$, where ${\mathbb U}$ is the
maximally extended Lorentzian surface which is depicted by the Carter--Penrose
diagram. Four-dimensional metric on ${\mathbb M}$ has the form (\ref{qnhtsd}), where
$d\Omega$ is the metric on sphere ${\mathbb S}^2$ for signature $(+---)$. If the signature
is opposite $(-+++)$, then we have to replace $Q^2\mapsto-Q^2$ and $\Lambda\to-\Lambda$
in the conformal
factor and change the sign of $d\Omega$ in metric (\ref{qnhtsd}). Due to the
existence of one Killing vector on Lorentzian surface ${\mathbb U}$, we are able to
classify all global solutions. To construct Carter--Penrose diagrams, we use the
conformal block method described in \cite{Katana00A} (see also,
\cite{Katana13B}). First, we consider
solutions of signature $(+---)$, and then with signature $(-+++)$.
\subsubsection{Metric signature $(+---)$}
If the metric signature is $(+---)$, then the conformal factor is
\begin{equation}                                                  \label{uvcfre}
  \Phi(q)=1-\frac{2M}q+\frac{Q^2}{q^2}-\frac{\Lambda q^2}3=:
  \frac{\varphi(q)+3Q^2}{3q^2},
\end{equation}
where we introduced the auxiliary function
\begin{equation}                                                  \label{uvxgij}
  \varphi(q):=-\Lambda q^4+3q^2-6Mq
\end{equation}
which is needed for further analysis. The case $Q=0$ was analyzed in
\cite{KaKlKu99}. Therefor, we classify solutions for $Q\ne0$. Without loss of
generality, we consider the case $Q>0$, because only $Q^2$ enters the conformal
factor.

Conformal factor (\ref{uvcfre}) has the second order pole $Q^2/q^2$ at zero and
the following asymptotic at infinity
\begin{equation*}
  \Phi\approx1-\frac{\Lambda q^2}3,\qquad q\to\infty.
\end{equation*}
If cosmological constant is equal to zero, then metric is asymptotically flat.
For $\Lambda>0$ and $\Lambda<0$, we have asymptotically de Sitter and anti-de Sitter
spacetime, respectively.

A global solution corresponds to one of the intervals $q\in(0,\infty)$ or
$q\in(-\infty,0)$ and $M>0$, because the curvature has singularity
(\ref{ubjdgh}) at zero, and space-time is not extendable through this point.
Roots of conformal factor (\ref{uvcfre}) correspond to horizons of space-time,
and Carter--Penrose diagrams are defined by the number and type of zeroes of the
conformal factor \cite{Katana00A}. Thus we have to analyse the number and type
of zeroes of conformal factor (\ref{uvcfre}) for all possible values of
constants $\Lambda$, $M\ge0$, and $Q>0$.

Note that conformal factor (\ref{uvcfre}) is invariant with respect to
transformation
\begin{equation*}
  M\to-M,\qquad q\to-q.
\end{equation*}
Therefor, instead of constructing global solutions on the interval
$q\in(0,\infty)$ for all values of $M$, we restrict ourselves only for
nonnegative $M\ge0$, but on two intervals $q\in(-\infty,0)$ and $(0,\infty)$.
This simplifies the analysis of the conformal factor.

We start with the simplest and well known case $\Lambda=0$.
\subsubsection{Metric signature $(+---)$. The case $\Lambda=0$.}
If cosmological constant vanishes, then zeroes of conformal factor
(\ref{uvcfre}) are defined by the quadratic equation
\begin{equation}                                                  \label{umdopi}
  q^2-2Mq+Q^2=0,
\end{equation}
which has two roots:
\begin{equation}                                                  \label{unbvgt}
  q_\pm=M\pm\sqrt{M^2-Q^2}.
\end{equation}

{\bf The Reissner--Nordstr\"om solution.}
For $Q<M$, there are two positive simple roots. This solution is called the
Reissner--Nordstr\"om solution \cite{Reissn16,Nordst18} and depicted by the
Carter--Penrose diagram S1 shown in Fig.\ref{fcarpenspherEM}. It was also found
by H.~Weyl \cite{Weyl17}. The solution has two horizons at $q_-$ and $q_+$ and
naked timelike singularity at $q=0$. The conformal factor tends to unity at
infinity, and, consequently, the Reissner--Nordstr\"om solution is
asymptotically flat. Arrows on the diagram show directions in which the solution
can be periodically extended in time. Instead of periodic extension, there is
the possibility to identify the opposite horizons. The singularity at $q=0$ is
timelike, and an observer can approach it as close as he likes in conformal
blocks I or III, and then enter universe III or I by going through conformal
block IV. Therefor, the Reissner--Nordstr\"om solution does not describe a
black hole.
\begin{figure}[p]
\hfill\includegraphics[width=0.9\textwidth]{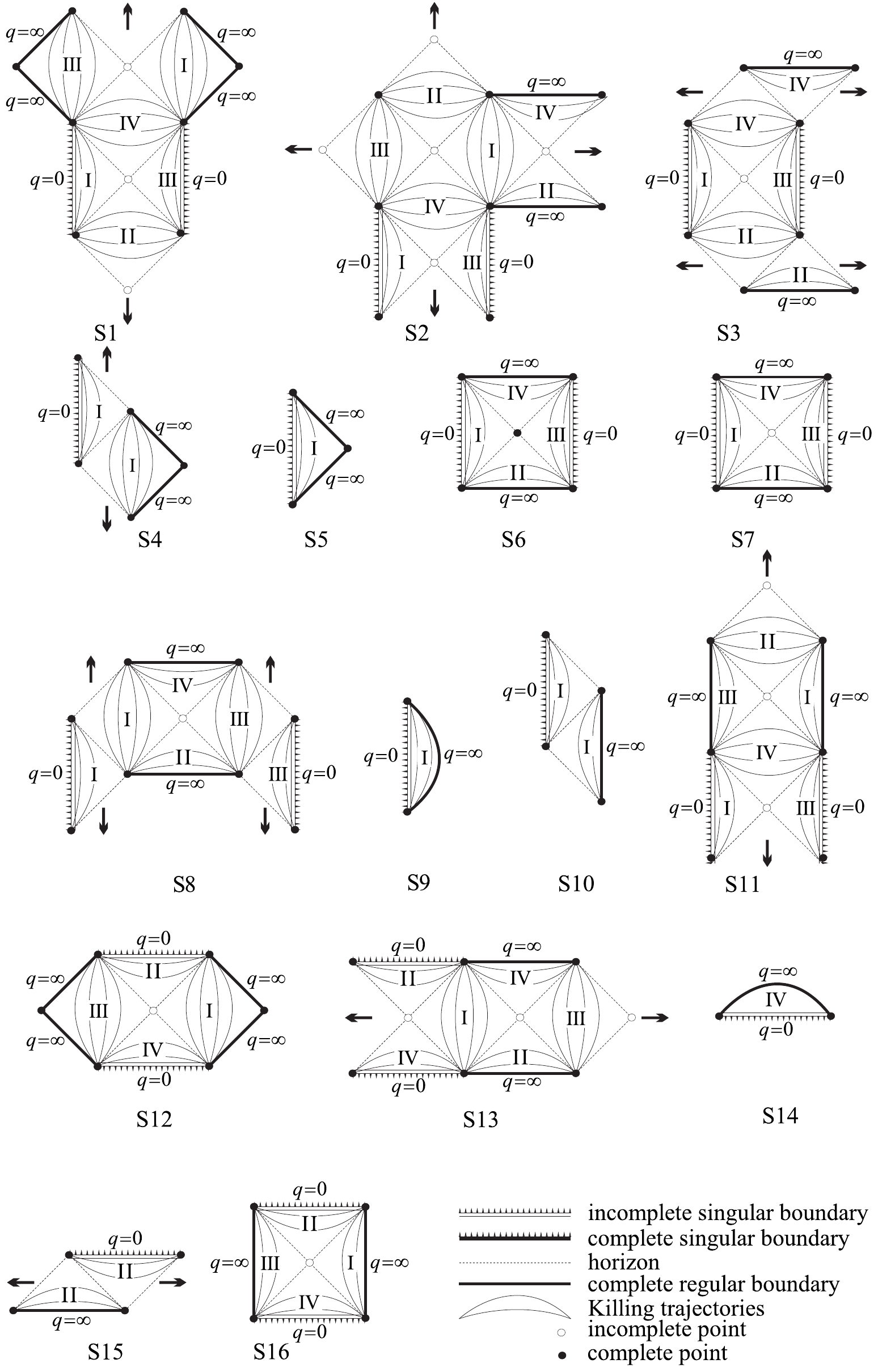}
\hfill {}
\centering\caption{The Carter--Penrose diagrams for spherically symmetric
solutions of Einstein's equations with electromagnetic field. Diagrams S1--S11
and S12--S16 correspond to metrics of signature $(+---)$ and $(-+++)$,
respectively.}
\label{fcarpenspherEM}
\end{figure}

{\bf Extremal black hole.}
For $Q=M$, the conformal factor is
\begin{equation*}
  \Phi=\frac{(q-M)^2}{q^2}.
\end{equation*}
It has one positive root of second order at $q=M$. The corresponding
Carter--Penrose diagram is shown in Fig.\ref{fcarpenspherEM}, S4. It is called
extremal black hole, though there is no any black hole since the singularity is
timelike and horizon surrounding the singularity is absent. There is also
space-reflected diagram.

{\bf Naked singularity.}
For $Q>M$, horizons are absent, and we have naked singularity shown in
Fig.\ref{fcarpenspherEM}, S5.  There is also space-reflected diagram.
\subsubsection{Metric signature $(+---)$. The case $\Lambda>0$.}
For positive cosmological constant, zeroes of the conformal factor are defined
by the fourth order equation
\begin{equation}                                                  \label{ubbcgt}
  \varphi(q)+3Q^2=0,
\end{equation}
where function $\varphi(q)$ is given by the fourth order polynomial (\ref{uvxgij}).
To draw Carter--Penrose diagrams, we do not need to know exact position of
zeroes. We have to know only their existence and type. Therefor, we analyze
function $\varphi(q)$ qualitatively and then move its graphic up, which corresponds
to increasing value of $Q^2$.

First, we differentiate function (\ref{ubbcgt}):
\begin{equation}                                                  \label{unmjqj}
\begin{split}
  \varphi'(q)=&-4\Lambda q^3+6q-6M,
\\
  \varphi''(q)=&-12\Lambda q^2+6=-6(2\Lambda q^2-1).
\end{split}
\end{equation}

The asymptotics of function $\varphi(q)$ ($\Lambda>0$) and its derivatives for $q=0$
and $q\to\infty$ are easily found:
\begin{align}                                                          \nonumber
  \varphi(0)=&~~0, & \varphi(q\to\infty)\approx&-\Lambda q^4,
\\                                                                \label{edswhy}
  \varphi'(0)=&-6M, & \varphi'(q\to\infty)\approx&-4\Lambda q^3,
\\                                                                     \nonumber
  \varphi''(0)=&~~6, & \varphi''(q\to\infty)\approx&-12\Lambda q^2.
\end{align}

Zeroes of function $\varphi(q)+3Q^2$ require more work. As we see later, their
number does not exceed three. To find the types of zeroes, we have to know local
extrema of function $\varphi(q)$, which become zeroes of order two or three after
shifting on $3Q^2$.

Local extrema of function $\varphi$ are defined by cubic equation (the solution is
given, i.e.\ in \cite{KorKor68})
\begin{equation}                                                  \label{uvskju}
  q^3-\frac3{2\Lambda} q+\frac{3M}{2\Lambda}=0.
\end{equation}
There are three qualitatively distinct cases depending on the value of constant
\begin{equation}                                                  \label{ebgtdj}
  \Upsilon:=-\frac1{8\Lambda^3}+\frac{9M^2}{16\Lambda^2}.
\end{equation}
Namely,
\begin{align*}
  \Upsilon>0\quad\Leftrightarrow\quad |M|>&\frac13\sqrt{\frac2\Lambda} \qquad
  \text
  {--\quad \parbox[c]{0.5\textwidth}{one real and two complex conjugate roots,}}
\\[6pt]
  \Upsilon=0\quad\Leftrightarrow\quad |M|=&\frac13\sqrt{\frac2\Lambda} \qquad
  \text{--\quad \parbox[c]{0.5\textwidth}{three real roots \\
  (at least two roots coincide),}}
\\[6pt]
  \Upsilon<0\quad\Leftrightarrow\quad |M|<&\frac13\sqrt{\frac2\Lambda} \qquad
  \text{--\quad three different real roots.}
\end{align*}

We start with the simplest case $\Upsilon=0$. This equality implies restriction
on ``mass'':
\begin{equation}                                                  \label{edbgdt}
  \Upsilon=0\qquad\Leftrightarrow\qquad M=\frac13\sqrt{\frac2\Lambda}.
\end{equation}
Moreover, roots of Eq.~(\ref{uvskju}) take the simple form:
\begin{equation}                                                  \label{unbhsk}
  M=\frac13\sqrt{\frac2\Lambda}:\qquad q_1=-\sqrt{\frac2\Lambda}, \qquad
  q_{2,3}=\frac12\sqrt{\frac2\Lambda},
\end{equation}
As we see, there are one simple negative root and one positive root of second
order for positive ``mass'' (\ref{edbgdt}).

If inequality $\Upsilon<0$ holds, then real roots of cubic equation
(\ref{uvskju}) are (see, i.e., \cite{KorKor68})
\begin{equation}                                                  \label{ubcndt}
  q_3=\sqrt{\frac2\Lambda}\cos\frac\alpha3,\qquad
  q_{2,1}=-\sqrt{\frac2\Lambda}\cos\left(\frac\alpha3\pm\frac\pi3\right),
\end{equation}
where
\begin{equation*}
  \cos\alpha:=-3M\sqrt{\frac\Lambda2}.
\end{equation*}
Since we consider only nonnegative $M$, then
$\alpha\in\big[\frac\pi2,\frac{3\pi}2\big]$. It implies existence of one negative
root $q_1$ and two positive: $q_2$ and $q_3$. We enumerate the zeroes in
Eq.~(\ref{ubcndt}) in such a way, that, in the limit
\begin{equation*}
  M\to\frac13\sqrt{\frac2\Lambda},
\end{equation*}
they take values (\ref{unbhsk}).

If $\Upsilon>0$, then we have only one negative root $q_1$. Its exact
position can be written but it is not needed.

Figure \ref{fgremone}, {\em a}, shows qualitative behavior of function $\varphi(q)$
for $\Lambda>0$ and different values of $M\ge0$.
\begin{figure}[hbt]
\hfill\includegraphics[width=.9\textwidth]{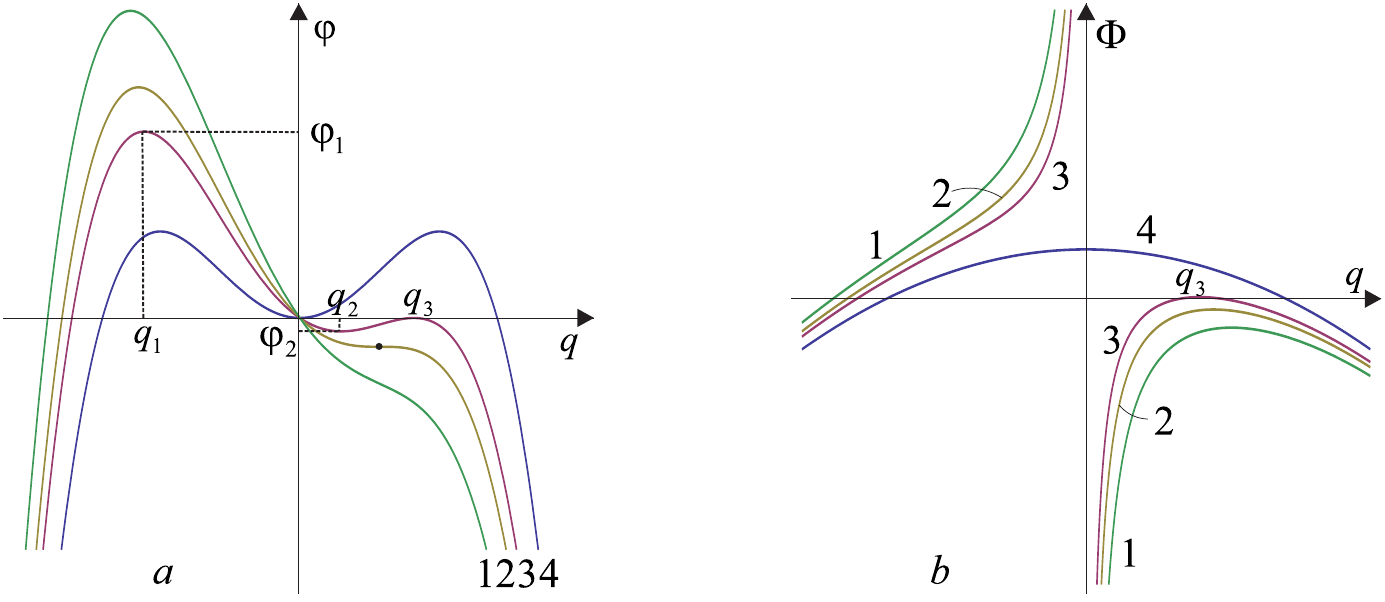}
\hfill {}
\centering\caption{Auxiliary function $\varphi(q)$ for $\Lambda>0$ {\em (a)} and
conformal factor $\Phi(q)$ for $Q=0$ {\em (b)}. The curves correspond to the
following values of the constant: (1) $\Upsilon>0$, (2) $\Upsilon=0$, (3)
$\Upsilon<0$, and (4) $\Upsilon=-\frac1{8\Lambda^3}$. Local extrema for curve 3
on the left picture are located at points $q_1$, $q_2$, and $q_3$. For curve
2, local maximum and minimum coincide, that is $q_2=q_3$, and are denoted by the
fat point. For curve 1, there is only one maximum for negative $q$. Curve 4 on
the left is symmetric with respect to substitution $q\mapsto-q$, has local
minimum at $q=0$, and two maxima at points $q=\pm\sqrt{\frac3{2\Lambda}}$.}
\label{fgremone}
\end{figure}
Now, to construct all global solutions which exist in the theory for signature
$(+---)$, we have to analyse zeroes of conformal factor $\Phi(q)$, qualitative
behavior of which for $Q=0$ is shown in Fig.\ref{fgremone}, {\it b}. Zeroes of
the conformal factor and their type coincide with that of function
$\varphi(q)+3Q^2$. Therefor, we have to shift up curves 1--4 in Fig.\ref{fgremone},
{\it a}, on $3Q^2$ to analyse its qualitative behavior. The number and type of
zeroes depend on curves 1--4 and on the value of the shift $3Q^2$. All possible
Carter--Penrose diagrams are drawn in Fig.\ref{fcarpenspherEM}.

The conformal factor depicted by curve 4 in Fig.\ref{fgremone}, {\it b}, does
not have zero at $q=0$. It corresponds to de Sitter space, and is degenerate at
this presentation of the problem ($M=0$, $Q=0$), which is not considered here
because of the assumption $Q>0$.

For qualitative description of behavior of the conformal factor, we introduce
notation:
\begin{equation}                                                  \label{ubvsik}
  \varphi_1:=\varphi(q_1),\qquad\varphi_2:=\varphi(q_2),\qquad\varphi_3:=\varphi(q_3),
\end{equation}
where $\varphi_1$ is the maximum, $\varphi_2$ is local minimum, and $\varphi_3$ is local
maximum of the auxiliary function $\varphi(q)$. One can easily verify, that, for
$\Lambda>0$ and $q<0$, the maximum is positive: $\varphi_1>0$. On positive half line
$q>0$, the local minimum is always negative: $\varphi_2<0$, and local maximum
$\varphi_3$ can take negative as well as positive values:
\begin{align*}
  0<M<\frac1{\sqrt\Lambda},& &&\varphi_3>0,
\\
  M=\frac1{\sqrt\Lambda}, & &&\varphi_3=0,
\\
  M>\frac1{\sqrt\Lambda}, & &&\varphi_3<0.
\end{align*}
When Eq.~(\ref{edbgdt}) holds, local minimum and maximum coincide: $q_2=q_3$.
Now we list all possibilities in the considered case.

{\bf Three horizons.} Under condition
\begin{equation}                                                  \label{ubcvdr}
  -\varphi_3<3Q^2<-\varphi_2,
\end{equation}
the conformal factor has three simple zeroes on positive half line. The
corresponding Carter--Penrose diagram of surface ${\mathbb U}$ is given by
S2 in Fig.\ref{fcarpenspherEM}. Here we have two timelike naked singularities.
Arrows show that this diagram can be either periodically continued in space- and
timelike directions, or opposite horizons can be identified. If we identify
horizons in one direction, them topologically the surface ${\mathbb U}$ is a cylinder.
If identification is performed in both directions, then it is a torus.

{\bf One simple horizon and timelike singularity.}
The conformal factor has one simple zero on positive half line under the
following conditions:
\begin{equation}                                                  \label{ubvgfg}
\begin{aligned}
  &\Lambda>0,\qquad& \Upsilon<0,& & {}\qquad&M=0,  & {}\qquad&{} & &Q\ne0,
\\
  &\Lambda>0,& \Upsilon<0,& & &M>0,  & &\varphi_3<0, & {}\qquad 3&Q^2<-\varphi_3,
\\
  &\Lambda>0,& \Upsilon<0,& & &M>0,  & &\forall\,\varphi_3, & 3&Q^2>-\varphi_2,
\\
  &\Lambda>0,& \Upsilon=0,& & &M>0,  & & & 3&Q^2\ne-\varphi_2,
\\
  &\Lambda>0,& \Upsilon>0,& & &M>0,  & & & &Q\ne0,
\\
  &\Lambda>0,& \forall\,\Upsilon,& & &M<0,  & &\forall\,\varphi_3, & &Q\ne0.
\end{aligned}
\end{equation}
This global solution is depicted by the Carter--Penrose diagram S7. It has
timelike singularity.

{\bf Triple horizon.} Under conditions:
\begin{equation}                                                  \label{uvbxfr}
  \Lambda>0,\qquad\Upsilon=0,\qquad M>0,\qquad 3Q^2=-\varphi_2.
\end{equation}
local maximum and minimum of auxiliary function $\varphi(q)$ coincide:
$q_2=q_3$, and the conformal factor has zero of third order at point $q_2$
(triple horizon). This case is depicted by diagram S6. It coincides with
diagram S7, but there is one important difference: the saddle point $q_2$ in the
center of the diagram is geodesically complete.

This diagram is interesting from physical standpoint. Consider a spacelike
section of this diagram. If the section does not go through the saddle point,
which is located in the center of the diagram, then it is an interval of finite
length with singular ends where two-dimensional curvature becomes infinite. If
the space section goes through the saddle point then it is the union of two
half-infinite intervals, because the central point in the center of the diagram
is the space infinity. If we introduce now global evolution parameter $T$, for
instance, vertical line on the diagram, then topology of space sections change
during evolution: for some value of $T$, there are two half-infinite intervals
instead of one finite interval. This example shows that changing topology of
space in time can occur already at the classical level. This type of diagram
appeared first in two-dimensional gravity with torsion \cite{Katana93A}.

{\bf Two horizons with double local minima.} Under conditions:
\begin{equation}                                                  \label{ubvcgt}
  \Lambda>0,\qquad\Upsilon<0,\qquad M>0,\qquad 3Q^2=-\varphi_2,
\end{equation}
the conformal factor has one zero of second order at point $q_2$ and one simple
simple zero at some point lying to the right from $q_2$. This solution is
depicted by Carter--Penrose diagram S8 with two timelike singularities, which
can be periodically extended in timelike direction.

{\bf Two horizons with double local maximum.} Under conditions:
\begin{equation}                                                  \label{ubvcgk}
  \Lambda>0,\qquad\Upsilon<0,\qquad M>0,\qquad \varphi_3<0,\qquad 3Q^2=-\varphi_3,
\end{equation}
the conformal factor has one double zero at $q_3$ and one simple zero at some
point lying to the left from $q_2$. This solution corresponds to Carter--Penrose
diagram S3 with two timelike singularities, which can be periodically extended
in spacelike direction.
\subsubsection{Metric signature $(+---)$. The case $\Lambda<0$.}
For negative cosmological constant, the conformal factor have the same form and
asymptotics remain the same (\ref{edswhy}). Equation (\ref{uvskju}) and
constant (\ref{ebgtdj}), defining the roots, do not change. We see that values
of constant $\Upsilon$ are positive for all $\Lambda$ and $M$. Consequently,
Eq.~(\ref{uvskju}) has only one nonnegative real root. Moreover, now branches of
auxiliary function $\varphi(q)$ are directed upwards as shown in Fig.\ref{fvfnelm},
and three new Carter--Penrose diagrams appear in the spherically symmetric case.
\begin{figure}[hbt]
\hfill\includegraphics[width=.9\textwidth]{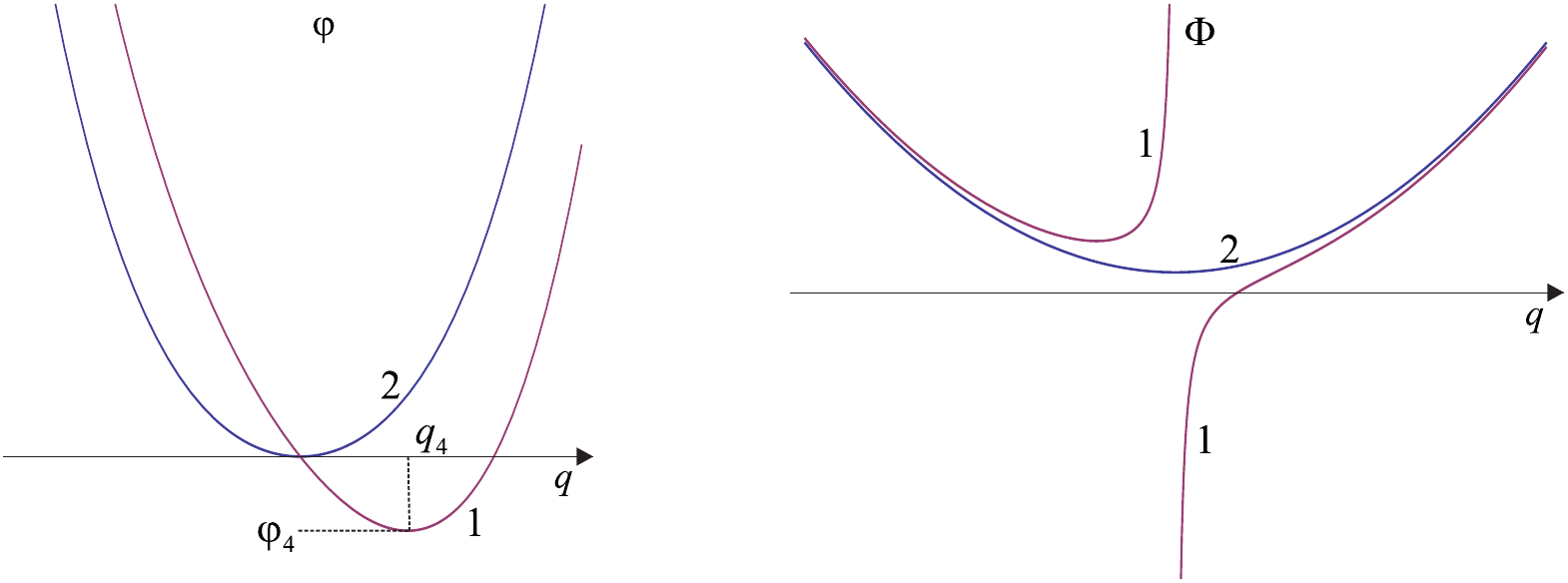}
\hfill {}
\centering\caption{Auxiliary function $\varphi(q)$ {\em (a)} and conformal factor
$\Phi(q)$ for $\Lambda<0$ and $Q=0$ {\em (b)}. The curves correspond to the
following values of the constant: (1) $\Upsilon>-\frac1{8\Lambda^3}$ and (2)
$\Upsilon=-\frac1{8\Lambda^3}$ ($M=0$). On the left picture, the only minimum of
curve 1 is located at point $q_4$. Curve 2 on the left is invariant with respect
to the map $q\mapsto-q$ and has minimum at $q=0$.}
\label{fvfnelm}
\end{figure}

The conformal factor depicted by curve 2 in Fig.\ref{fvfnelm}, {\it b}, has zero
at point $q=0$. It corresponds to anti-de Sitter space and is the degenerate
case in the problem under consideration ($M=0$, $Q=0$).

Now we list all possibilities for negative cosmological constant.

{\bf Timelike singularity.} Under conditions:
\begin{equation}                                                  \label{ubvgfl}
\begin{aligned}
  &\Lambda<0,\qquad& {}\qquad&M>0,  & {}\qquad&{} & &3Q^2>-\varphi_4,
\\
  &\Lambda<0,&  &M\le0,  & &\varphi_3<0, & {}\qquad 3&Q\ne0,
\end{aligned}
\end{equation}
the conformal factor does not have zeroes, and, consequently, horizons are
absent. In this case, the Carter--Penrose diagram has the lens form S9 in
Fig.\ref{fcarpenspherEM}. There is also space-reflected diagram.

{\bf Naked singularity.} Under conditions:
\begin{equation}                                                  \label{ubndht}
  \Lambda<0,\qquad M>0,\qquad 3Q^2=-\varphi_4,
\end{equation}
the conformal factor has one positive root of second order at the minimum of
the auxiliary function at $q_4$. In this case, the Carter--Penrose diagram is
S10 in Fig.\ref{fcarpenspherEM}. In contrast to the naked singularity S4, the
right complete infinity $q=\infty$ is timelike. It is due to asymptotic of the
conformal factor at infinity, because space-time is asymptotically
anti-de Sitter for $\Lambda<0$. There is also space-reflected diagram.

{\bf Timelike singularity and two horizons.} Under conditions:
\begin{equation}                                                  \label{ubvcfj}
  \Lambda<0,\qquad M>0,\qquad 3Q^2<-\varphi_4,
\end{equation}
the conformal factor has two zeroes. In this case, the Carter--Penrose diagram
is given by S11 in Fig.\ref{fcarpenspherEM}. This solution can either be
periodically extended in timelike direction or opposite horizons can be
identified. In contrast to diagram S1, space infinities are timelike, which is
due to asymptotic at infinity.

Thus we classified all spherically symmetric global solutions of Einstein's
equations with electromagnetic field for metric signature $(+---)$. We see, that
all solutions of signature $(+---)$ contain timelike singularity.
Totally, we get 11 topologically inequivalent solutions S1--S11. It is possible
to give more subtle classification taking into account existence of
degenerate and oscillating geodesics. The latter appears, if the conformal
factor has local extremum inside one of the conformal blocks. This
classification was given for global solutions of two-dimensional gravity with
torsion \cite{Katana93A}.
\subsubsection{Metric signature $(-+++)$}
If the signature is opposite, the conformal factor has the form (\ref{uvcfre})
but with the replacement $Q^2\mapsto-Q^2$. It means that auxiliary function
$\varphi(q)$ in Figs.~\ref{fgremone} and \ref{fvfnelm}, {\it a}, remains the same,
but we have to move it downwards instead of upwards. There are 5 new
Carter--Penrose diagrams.

We start with the simplest case.
\subsubsection{Metric signature $(-+++)$. The case $\Lambda=0$.}
In the considered case, zeroes of the conformal factor are defined by quadratic
equation
\begin{equation*}
  q^2-2Mq-Q^2=0,
\end{equation*}
which has two roots:
\begin{equation*}
  q_\pm=M\pm\sqrt{M^2+Q^2}.
\end{equation*}
It implies inequalities $q_+>0$ and $q_-<0$ for $Q>0$. Therefor, there is one
simple horizon for any $M$. Consequently, the Carter--Penrose diagram has
exactly the same form as for Schwarzschild black hole S12 in
Fig.\ref{fcarpenspherEM}.
\subsubsection{Metric signature $(-+++)$. The case $\Lambda>0$.}
Auxiliary function $\varphi(q)$ is the same (\ref{uvxgij}), but it has to be moved
on $3Q^2$ downwards. For positive cosmological constant, the qualitative
behavior of the auxiliary function is shown in Fig.\ref{fgremone}, {\em a}.

{\bf Spacelike singularity.} Under conditions:
\begin{equation}                                                  \label{ubvgfj}
\begin{aligned}
&\Lambda>0,\qquad& &\Upsilon<0,& {}\qquad&M>0,  & {}\qquad\varphi_3\le0,&{} &&Q\ne0,
\\
  &\Lambda>0,& &\Upsilon<0,& &M>0,  & \varphi_3>0,& & {}\qquad 3&Q^2>\varphi_3,
\\
  &\Lambda>0,& &\forall\Upsilon,& &M<0,  & \varphi_3>0,& & {}\qquad 3&Q^2>\varphi_1,
\end{aligned}
\end{equation}
the conformal factor does not have roots. In this case, there is spacelike
singularity without horizons. Its Carter--Penrose diagram is S14 in
Fig.~\ref{fcarpenspherEM}. There is also time-reflected diagram.

{\bf Spacelike singularity with two horizons.} Under conditions:
\begin{equation}                                                  \label{ubvgff}
\begin{aligned}
  &\Lambda>0,\qquad& \Upsilon<0,&{}\qquad&M>0,&\qquad\varphi_3>0,& &  0<3Q^2<\varphi_3,
\\
  &\Lambda>0,& \forall\,\Upsilon,& &M<0,  & &  {}\qquad &0<3Q^2<\varphi_1,
\end{aligned}
\end{equation}
the conformal factor has two simple zeroes, and, consequently, two horizons.
Moreover the singularity at $q=0$ is spacelike. This solution is depicted by
diagram S13 in Fig.~\ref{fcarpenspherEM}. It can be either periodically extended
in spacelike direction, or we can identify the opposite horizons. This solution
describes white and black holes, which are periodically located in spacelike
directions. Moreover, if an observer is located in the domain IV, he has the
opportunity either to live forever, or to fall on one of two black holes.

{\bf Spacelike singularity with one double horizon.} Under conditions:
\begin{equation}                                                  \label{ubvgfk}
\begin{aligned}
  &\Lambda>0,\qquad& \Upsilon<0,&{}\qquad&M>0,&\qquad\varphi_3>0,& &  3Q^2=\varphi_3,
\\
  &\Lambda>0,& \forall\,\Upsilon,& &M<0,  & &  {}\qquad &3Q^2=\varphi_1,
\end{aligned}
\end{equation}
the conformal factor has one double zero, and the singularity is spacelike. This
global solution is given by the Carter--Penrose diagram S15 in
Fig.~\ref{fcarpenspherEM}, which can be periodically extended in spacelike
direction. This solution describes the collection of black and white (after
time reflection) holes. As in the previous case, an observer in domain II has
the choice either to live forever or to fall on one of two black holes. There is
also time-reflected diagram.
\subsubsection{Metric signature $(-+++)$. The case $\Lambda<0$.}
For negative cosmological constant and signature $(-+++)$, the auxiliary
function has previous form and is shown in Fig.\ref{fvfnelm}. To find zeroes,
its graphic must be moved downwards. Thus for all values of parameters:
\begin{equation}                                                  \label{unbhju}
  \Lambda<0,\qquad\forall\,\Upsilon,\qquad\forall M,\qquad Q\ne0,
\end{equation}
it has one simple zero. This global solution is given by diagram S16 in
Fig.~\ref{fcarpenspherEM}. In this case we have asymptotically anti-de Sitter
black hole. Note, that, for positive mass, the space-time has degenerate and
oscillating geodesics, because local minimum exists for $q>0$. For $M<0$ these
geodesics are absent.

Thus, for metric signature $(-+++)$, there are only 5 topologically different
global solutions S12--S16. All singularities in this case are spacelike and
correspond either to black or white holes.
\subsection{Planar solutions $K^h=0$                             \label{sdbhft}}
If Gaussian curvature of surface ${\mathbb V}$ equals to zero, then it is either the
Euclidean plane ${\mathbb R}^2$, or a cylinder, or a torus (after factorization).
Thus, there is spontaneous ${\mathbb I}{\mathbb S}{\mathbb O}(2)$ symmetry arising if the surface
${\mathbb V}$ is Euclidean plane ${\mathbb R}^2$. That is, the space-time metric becomes
invariant with respect to ${\mathbb I}{\mathbb S}{\mathbb O}(2)$ transformation group on the equations
of motion. In Schwarzschild coordinates $(\zeta,q,y,z)$, it is written in the form
(for $m=-q^2<0$, corresponding to signature $(+---)$):
\begin{equation}                                                  \label{uvxcfd}
  ds^2=\Phi(q)d\zeta^2-\frac{dq^2}{\Phi(q)}-q^2d\Omega_{\textsc{p}},
\end{equation}
where
\begin{equation}                                                  \label{ubcvdf}
  \Phi(q)=-\frac{2M}q+\frac{Q^2}{q^2}-\frac{\Lambda q^2}3,\qquad
  d\Omega_{\textsc{p}}:=dy^2+dz^2.
\end{equation}

To draw Carter--Penrose diagrams for Lorentzian surface ${\mathbb U}$, we have to
analyse zeroes and asymptotics of conformal factor $\Phi(q)$. For $Q\ne0$,
we have the second order pole $Q^2/q^2$ at zero and asymptotic at infinity
\begin{equation*}
  \Phi\approx-\frac{\Lambda q^2}3,\qquad q\to\infty
\end{equation*}
On intervals $(0,\infty)$ and $(-\infty,0)$, the conformal factor is smooth,
and, consequently, every global solution corresponding to one of these
intervals is smooth.
As for spherically symmetric solutions, we consider positive $M$ on both
intervals due to the symmetry transformation $(M,q)\mapsto(-M,-q)$.

We start with the simplest case.
\subsubsection{Metric signature $(+---)$. The case $\Lambda=0$.}
The conformal factor is
\begin{equation}                                                  \label{ubjsuu}
  \Phi(q)=\frac{Q^2-2Mq}{q^2}.
\end{equation}
It has obviously one simple zero
\begin{equation*}
  q=\frac{Q^2}{2M}.
\end{equation*}
Moreover, there are only two cases.

{\bf Timelike singularity and one horizon.}
Under conditions:
\begin{equation}                                                  \label{ubcvdi}
  \Lambda=0,\qquad M>0,
\end{equation}
the conformal factor has one simple positive zero. The corresponding
Carter--Penrose diagram is P1 in Fig.~\ref{fcarpenpl}. This diagram has the
same form as the Schwarzschild black hole S12 but turned over on $90^\circ$.
\begin{figure}[hbt]
\hfill\includegraphics[width=.3\textwidth]{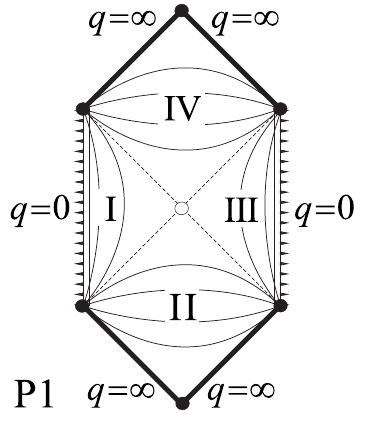}
\hfill {}
\centering\caption{The Carter--Penrose diagram for planar solution for
$\Lambda=0$ and $M>0$.}
\label{fcarpenpl}
\end{figure}

{\bf Naked singularity.}
Under conditions:
\begin{equation}                                                  \label{ubvcgf}
  \Lambda=0,\qquad M\le0,
\end{equation}
positive roots of the conformal factor are absent, and we have naked singularity
S5 in Fig.~\ref{fcarpenspherEM}.
\qed

To find zeroes for nonzero cosmological constant $\Lambda\ne0$, we introduce
auxiliary function $\phi(q)$ representing the conformal factor for signature
$(+---)$ in the form
\begin{equation}                                                  \label{ubvuiu}
  \Phi(q)=:\frac{\phi(q)+3Q^2}{3q^2},
\end{equation}
where
\begin{equation}                                                  \label{uvxbsf}
  \phi(q):=-6Mq-\Lambda q^4.
\end{equation}
For the opposite signature, ${\sf\,sign\,}\widehat g=(-+++)$, it is needed to make
replacement $Q^2\mapsto-Q^2$. We see that on intervals $(0,\infty)$ and
$(-\infty,0)$ the number and type of zeroes of the conformal factor coincide
with zeroes of numerator $\phi(q)+3Q^2$. It means that auxiliary function must
be shifted either downwards (signature $(+---)$), or upwards (signature
$(-+++)$).

Auxiliary function (\ref{uvxbsf}) has two real roots:
\begin{equation*}
  q=0,\qquad q=\sqrt[3]{-\frac{6M}\Lambda},
\end{equation*}
and two complex conjugate roots which do not interest us. Qualitative behavior
of the auxiliary function and corresponding conformal factor are shown in
Fig.~\ref{fvplanar}. Position of extrema of the auxiliary function is defined by
the equality
\begin{equation*}
  \phi'(q)=-6M-4\Lambda q^3=0\qquad\Rightarrow\qquad q=\sqrt[3]{-\frac{3M}{2\Lambda}}.
\end{equation*}
We denote them by $q_5$ and $q_6$ for $\Lambda>0$ and $\Lambda<0$, respectively
(see.\ Fig.~\ref{fvplanar}, {\it a}).
\begin{figure}[hbt]
\hfill\includegraphics[width=.9\textwidth]{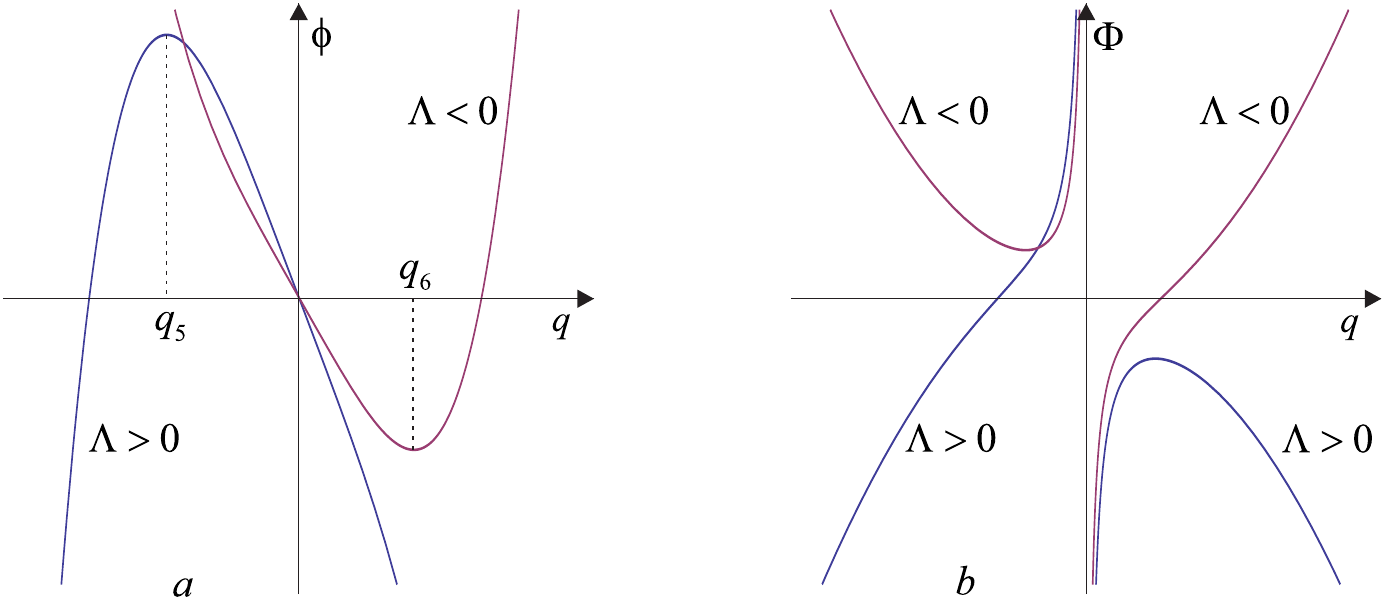}
\hfill {}
\centering\caption{Auxiliary function $\phi(q)$ {\em (a)} and conformal factor
$\Phi(q)$ {\em (b)} при $M>0$. Maximum and minimum of the auxiliary function
are located ap points $q_5$ and $q_6$ for $\Lambda>0$ and $\Lambda<0$, respectively.}
\label{fvplanar}
\end{figure}
The maximal and minimal values of the auxiliary function are denoted by
\begin{equation*}
  \phi_{5,6}:=\phi(q_{5,6})=\frac92M\sqrt[3]{\frac{3M}{2\Lambda}}.
\end{equation*}
It is clear, that $\phi_5>0$ for $\Lambda>0$ and $\phi_6<0$ for $\Lambda<0$.

Detailed analysis show that Carter--Penrose diagrams for all planar solutions
for $\Lambda\ne0$ were already met in the spherically symmetric case. Therefor,
to save space, we give classification of all planar solutions in table
\ref{tplanar}.
\begin{table}[htb]
\begin{center}
  \begin{tabular}{|c|c|c|c|l|}                                            \hline
    $+---$ & $\Lambda>0$ & $\forall M$ & $Q\ne0$ & S7 \\ \hline
    $+---$ & $\Lambda<0$ & $M>0$ & $0<3Q^2<-\phi_6$ & S11 \\ \hline
    $+---$ & $\Lambda<0$ & $M>0$ & $3Q^2=-\phi_6$ & S10 \\ \hline
    $+---$ & $\Lambda<0$ & $M>0$ & $3Q^2>-\phi_6$ & S9 \\ \hline
    $+---$ & $\Lambda<0$ & $M\le0$ & $Q\ne0$ & S9 \\ \hline
    $-+++$ & $\Lambda>0$ & $M\ge0$ & $Q\ne0$ & S14 \\ \hline
    $-+++$ & $\Lambda>0$ & $M<0$ & $0<3Q^2<\phi_5$ & S13 \\ \hline
    $-+++$ & $\Lambda>0$ & $M<0$ & $3Q^2=\phi_5$ & S15 \\ \hline
    $-+++$ & $\Lambda>0$ & $M<0$ & $3Q^2>\phi_5$ & S14 \\ \hline
    $-+++$ & $\Lambda<0$ & $\forall M$ & $Q\ne0$ & S16 \\ \hline
    \end{tabular}
    \caption{\label{tplanar} Classification of global planar solutions
    for $\Lambda\ne0$.}
\end{center}
\end{table}
Note, that diagrams S7, S9, S10 and S11 differ from diagrams S16, S14, S15 and
S13 by the turn on $90^\circ$ degrees, respectively.
\section{Hyperbolic global solutions                             \label{sbcvdg}}
If Gaussian curvature of surface ${\mathbb V}$ is negative, $K^h=-1$, then the surface
is two-sheeted hyperboloid ${\mathbb H}^2$, more precisely, the upper sheet of
two-sheeted hyperboloid (the Lobachevsky plane). It is the universal covering
surface for closed Riemannian surfaces of genus two and higher. If ${\mathbb V}={\mathbb H}^2$,
then the isometry group is the Lorentz group ${\mathbb S}{\mathbb O}(1,2)$. In this case, the
metric in Schwarzschild coordinates $(\zeta,q,\theta,\varphi)$ for signature $(+---)$
has the form
\begin{equation}                                                  \label{uvxbif}
  ds^2=\Phi(q)d\zeta^2-\frac{dq^2}{\Phi(q)}-q^2d\Omega_{\textsc{h}},
\end{equation}
where
\begin{equation*}
  \Phi(q)=-1-\frac{2M}q+\frac{Q^2}{q^2}-\frac{\Lambda q^2}3,\qquad
  d\Omega_{\textsc{h}}:=d\theta^2+{\sf\,sh\,}^2\theta d\varphi^2.
\end{equation*}
The conformal factor for this metric differs from that in the spherically
symmetric case (\ref{uvcfre}) by the transformation
\begin{equation}                                                  \label{unvbfg}
  \Phi\mapsto-\Phi,\qquad M\mapsto-M,\qquad Q^2\mapsto-Q^2,\qquad\Lambda\mapsto-\Lambda.
\end{equation}
In addition, transformation $Q^2\mapsto-Q^2$ corresponds to signature change of
the metric, $(+---)\mapsto(-+++)$. Since we have already described global
spherically symmetric solutions for all values of $M,Q^2$ and $\Lambda$, all
hyperbolic solutions are obtained from spherically symmetric ones by simple
rotation of Carter--Penrose diagrams by $90^\circ$, which corresponds to
transformation $\Phi\mapsto-\Phi$. In this way we get 16 additional
Carter--Penrose diagrams.
\section{Conclusion}
We assumed that four-dimensional space-time is the warped product of two
surfaces, ${\mathbb M}={\mathbb U}\times{\mathbb V}$, and find a general solution of Einstein's
equations with cosmological constant and electromagnetic field. These solutions
are well known locally and partly globally. We give classification of all global
solutions in the case when surface ${\mathbb V}$ is of constant curvature. Totally,
there are 37 topologically different global solutions.
These solutions in case {\sf B} have four Killing vector fields, three
of them corresponding to symmetry of the metric on constant curvature surface
${\mathbb V}$. They are generators of isometry groups ${\mathbb S}{\mathbb O}(3)$, ${\mathbb I}{\mathbb S}{\mathbb O}(2)$, and
${\mathbb S}{\mathbb O}(1,2)$ in cases when surface ${\mathbb V}$ is a sphere ${\mathbb S}^2$, Euclidean plane
${\mathbb R}^2$, and two-sheeted hyperboloid ${\mathbb H}^2$, respectively. The fourth Killing
vector
generalizes Birkhoff's theorem. In all cases, there is ``spontaneous symmetry
emergence'' because the existence of Killing vector fields was not assumed at
the beginning, and their appearance is the consequence of Einstein's equations.
Most probably, part of the constructed solutions are not satisfactory from
physical point of view. For example, for given signs in the Lagrangian and
signature of the metric $(-+++)$, the Carter--Penrose diagram for charged
black hole coincide with the Schwarzschild solution. However, the quadratic form
of momenta in the canonical Hamiltonian for physical degrees of freedom is not
positive definite (ghosts appearance), and this solution have to be discarded as
unphysical. Nevertheless, the given classification of global solutions of
Einstein's equations in the form of warped product of two surfaces is important,
because we must know what is to be discarded.

\end{document}